\documentclass[12pt,a4paper]{article}

\usepackage{amsmath,amsthm,amssymb,amscd,a4wide}

\usepackage{dsfont}
\usepackage{mathrsfs}
\usepackage{setspace}
\usepackage{url}
\usepackage{color}

\usepackage[T1]{fontenc}
\usepackage[latin1]{inputenc}

\usepackage{graphicx}
\usepackage[bf]{caption2}

\numberwithin{equation}{section}

\DeclareMathOperator{\op}{op}

\newcommand{\ud}{\mathrm{d}}
\newcommand{\ui}{\mathrm{i}}
\newcommand{\ue}{\mathrm{e}}
\newcommand{\be}{\begin{equation}}
\newcommand{\ee}{\end{equation}}

\DeclareMathOperator{\T}{T}

\DeclareMathOperator{\arsinh}{arsinh}

\newcommand{\gz}{\mathbb Z}

\newcommand{\rz}{\mathbb R}
\newcommand{\nz}{\mathbb N}
\newcommand{\kz}{\mathbb C}

\newcommand{\ba}{\begin{aligned}}
\newcommand{\ea}{\end{aligned}}

\newtheorem{lemma}{Lemma}
\newtheorem{prop}{Proposition}

\begin{document}

\thispagestyle{empty}

\begin{center}

{\LARGE\bf The Berry-Keating operator on a lattice}\\

\vspace*{3cm}

{\large Jens Bolte\footnote{Department of Mathematics, Royal Holloway,
University of London, Egham, TW20 0EX, United Kingdom, 
{\tt jens.bolte@rhul.ac.uk}}, 
{\large Sebastian Egger}\footnote{Department of Mathematics,
Technion--Israel Institute of Technology, 629 Amado Building, Haifa 32000, 
Israel, {\tt egger@tx.technion.ac.il }}, and 
Stefan Keppeler\footnote{Fachbereich Mathematik, Universit\"at T\"ubingen, 
Auf der Morgenstelle 10, 72076 T\"ubingen, Germany, 
{\tt stefan.keppeler@uni-tuebingen.de}}}
\end{center}
\vfill
\begin{abstract}
  We construct and study a version of the Berry-Keating operator corresponding
  to a classical Hamiltonian on a compact phase space, which we choose to be 
  a two-dimensional torus. The operator is a Weyl quantisation of the classical
  Hamiltonian for an inverted harmonic oscillator, producing a
  difference operator on a finite, periodic lattice. We investigate
  the continuum and the infinite-volume limit of our model in
  conjunction with the semiclassical limit. Using semiclassical
  methods, we show that only a specific combination of the limits
  leads to a logarithmic mean spectral density as it was anticipated
  by Berry and Keating.
\end{abstract}

\newpage

\section{Introduction}
Phase-space methods in quantum mechanics are often used in a
semiclassical context because they link (pseudo-)differential
operators and functions on a classical phase space, see, e.g.,
\cite{Zworski:2012}. In such a context the phase space is a cotangent
bundle of a manifold (the configuration space) and, hence, is not
compact. In contrast, quantisations on compact phase spaces lead to
operators in finite dimensional Hilbert spaces. The latter approach is
frequently used to quantise symplectic maps on tori
\cite{Bouzouina:1996,Esposti:2003}, the most prominent example being
the cat map \cite{Hannay:1980}.  However, one can also quantise
Hamiltonian flows on tori and therefore obtain a phase-space
representation of a Schr\"odinger equation in a finite dimensional
Hilbert space. In this context the Hamiltonian operator can be
interpreted as a difference operator on a finite lattice with periodic
boundary conditions. If the classical phase space is a two-dimensional
torus one has two length parameters available in addition to the
semiclassical parameter.  In models with difference operators on
periodic lattices these length scales can be used to perform a
continuum limit as well as an infinite volume limit.

In this paper we explore the various limits in a discrete variant of
the Berry-Keating operator that was introduced in
\cite{BerKea99a}. The latter is a differential operator with a
self-adjoint realisation in $L^2(\rz)$ and, therefore, can be
represented in the phase space $\T^\ast\!\rz$. Inspired by the work of
Connes \cite{Con96,Con99}, Berry and Keating intended to define a
self-adjoint operator whose spectrum is related to the non-trivial
zeros of the Riemann zeta function. In particular, the mean spectral
density of this operator should follow the same asymptotic,
logarithmic law as that of the Riemann zeros. To this end they
considered a quantisation of the classical Hamiltonian
$H(q,p)=qp$. The first obstacle that needed to be addressed was that
the naive version of the operator has a purely continuous spectrum,
which is related to the fact that the energy surfaces in phase space
are not compact. For that reason Berry and Keating suggested a
truncation of these energy surfaces \cite{BerKea99a,BerKea99b},
however, without defining an operator that would correspond to that
truncation. Using standard semiclassical techniques, they then showed
that the truncated energy surfaces would lead to the desired form of
the mean spectral density. Later, a modification was suggested
\cite{SieRod11} to overcome the problems with the phase-space
truncation, and this was subsequently improved in \cite{BerKea11}.  
These modifications keep the non-compact phase space $\T^\ast\!\rz$,
but add terms to the classical Hamiltonian in such a way that its energy
surfaces are compact, yet the resulting semiclassical spectral density
has the same leading asymptotic behaviour as before. After quantisation, 
however, the additional terms lead to non-local contributions to the
quantum Hamiltonian.

In other approaches two-dimensional, semi-inverted oscillators were
related to the Riemann zeta function and its zero-count
\cite{BhaKhaReiTom97}, and versions of the Berry-Keating operator on a
half-line as well as on quantum graphs were investigated
\cite{SteEnd10}. The connection of the Berry-Keating operator with
local versions of the Riemann hypothesis was studied in
\cite{Sre11}. More recently, a PT-symmetric, non-hermitian variant of
the Berry-Keating operator \cite{BenBroMue16} as well as a version in
polymeric quantum mechanics \cite{BerMol16} were studied.

Our proposal differs from previous approaches in that we use a
two-dimensional torus as phase space, but keep the classical
Hamiltonian (up to a linear canonical transformation). Weyl
quantisation on this compact phase space will immediately produce a
self-adjoint operator in an $N$-dimensional Hilbert space that,
therefore, will have a discrete, even finite, spectrum.  The number
$N$ of eigenvalues tends to infinity in the semiclassical limit, which
in this context is well known to be the limit $N\to\infty$. The model
allows us to study various limits of the operator and its spectrum,
including semiclassical, continuum and infinite-volume limits, or
combinations thereof. We also use semiclassical methods that are
mainly based on a Gutzwiller trace formula which we proved previously
\cite{BEK16}. In that context we identify a way to choose the length
scales of the torus in such a way that they produce the same cut-off
as it was imposed in \cite{BerKea99a}. As a consequence, the
semiclassical approximation of the spectral density is the same as in
\cite{BerKea99a}, except for a factor of two whose origin we explain in
Section~\ref{sec:semiclassics}.

The paper is organised as follows. In Section~\ref{sec:mod} we review
the model of Berry and Keating and then introduce our model. In the
next section we analyse the continuum and the infinite-volume
limit. This is first done very explicitly in the example of the
quantisation of the symbol $\xi^2$, and then qualitatively for our
model. Section~\ref{sec:semiclassics} is devoted to a semiclassical
calculation of the spectral density in our model. The result is
compared to numerical data in the following section. We present our
conclusions in Section~\ref{sec:conclusions}. Two appendices contain a
proof of a lemma in Section~\ref{sec:mod} and the calculation of the
matrix elements of our model operator, respectively.

\section{The model}
\label{sec:mod}
Berry and Keating \cite{BerKea99a,BerKea99b} consider the operator
\begin{equation}
\label{eq:BKoperator}
 H_{\rm BK} = \frac{\hbar}{\ui}x\frac{\ud}{\ud x}+\frac{\hbar}{2\ui},
\end{equation}
which is a (Weyl) quantisation of the symbol $h_{\rm BK}(q,p)=qp$
on the phase space $\T^\ast\!\rz$. Their intention is to count the 
number of eigenvalues in spectral intervals,
however, the operator $H_{\rm BK} $ defined on a suitable domain in
$L^2(\rz)$ has a continuous spectrum and no eigenvalues. They, therefore,
suggest a cut-off in phase space to enforce the energy surfaces
$h_{\rm BK}^{-1}(E)\subset\T^\ast\!\rz$ to be compact. One would expect a
modification of the operator \eqref{eq:BKoperator} that reflects the
truncation of the energy surfaces to possess a purely discrete
spectrum. However, no definition of an operator is given.

\begin{figure}[t]
\begin{center}
\resizebox{.32\textwidth}{!}{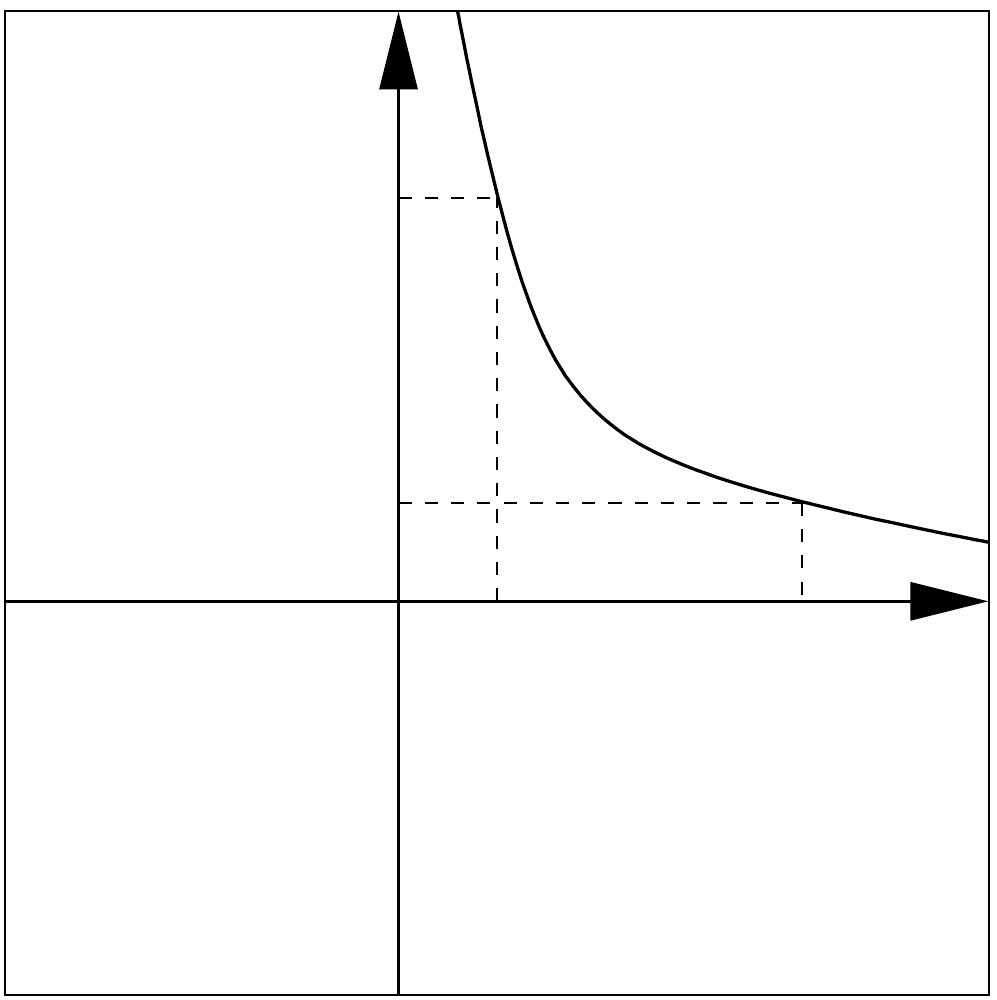}
\resizebox{.32\textwidth}{!}{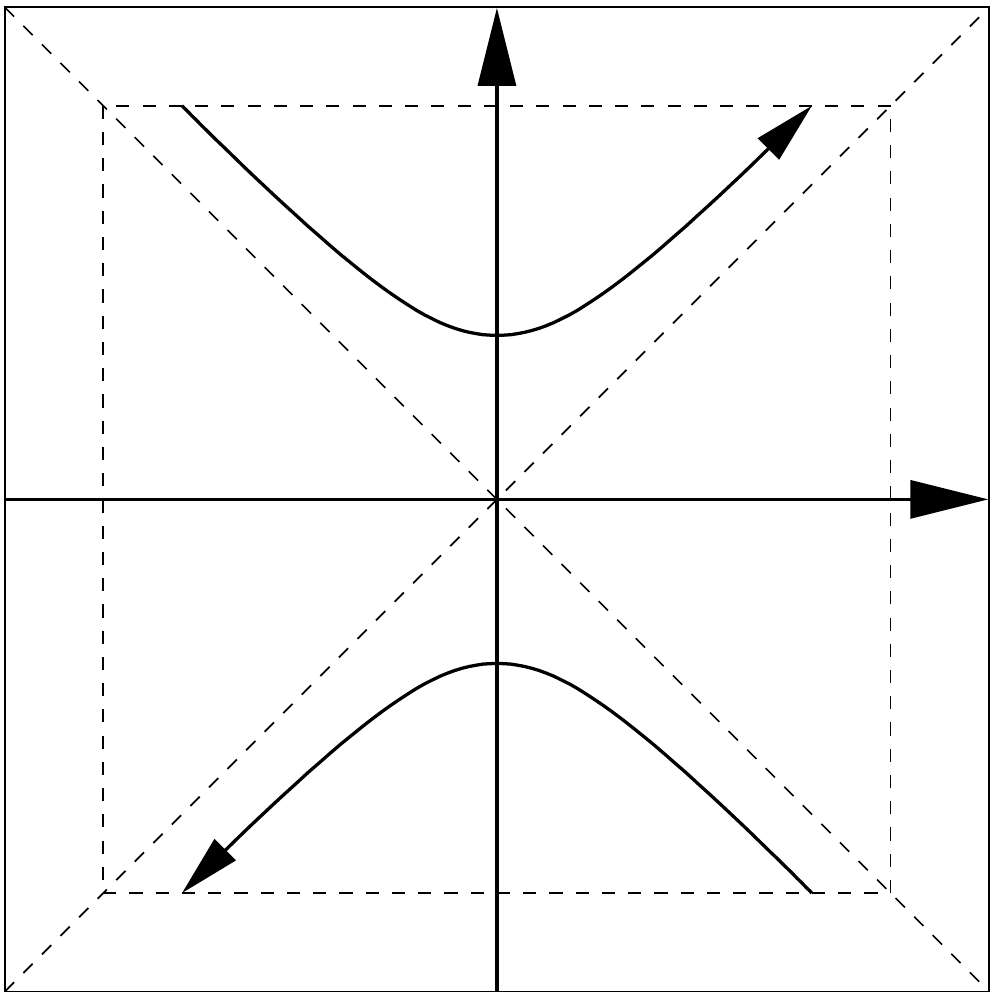}
\resizebox{.32\textwidth}{!}{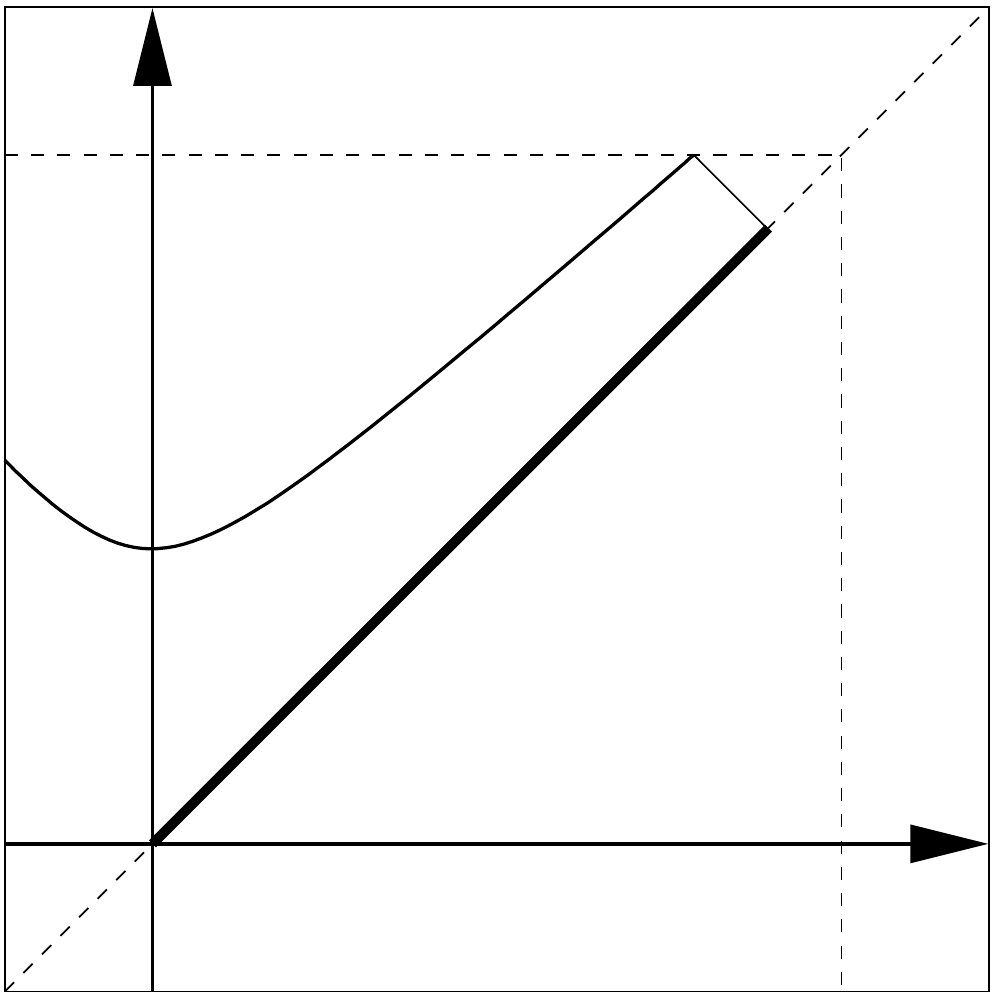}
\caption{Energy surfaces (a) for the classical Hamiltonian
  corresponding to the Berry-Keating operator, with cut-off parameters
  $L_q$ and $L_p$, and (b) for the rotated Hamiltonian
  \eqref{eq:classhamilton}, with fundamental domain characterised by
  $\ell_x$ and $\ell_\xi$; (c) relation between $L_{q,p}$ and
  $\ell_{x,\xi}$, see Sec.~\ref{sec:semiclassics}.}
\label{fig:energy_contours}
\end{center}
\end{figure}

The cut-off procedure they suggest is to only consider the subset
\begin{equation}
\label{eq:pstrunc}
  \{(q,p)\in\T^\ast\!\rz;\ q\geq L_q,\ p\geq L_p\},
\end{equation}
of phase space, where $L_q>0$ and $L_p>0$ are parameters satisfying
$L_q L_p=2\pi\hbar$. Before the truncation the energy surface consists
of the points $(q,p)\in\T^\ast\!\rz$ on the two branches of the
hyperbola $qp=E$.  After the truncation only the part in the first
quadrant satisfying
\begin{equation}
L_q\leq q\leq\frac{E}{L_p}\quad\text{and}\quad L_p\leq p\leq\frac{E}{L_q}
\end{equation}
remains, see Figure~\ref{fig:energy_contours}(a). This has finite
volume, suggesting that a related operator would possess a purely
discrete spectrum. Berry and Keating also propose that one should
identify $p$ with $-p$ and/or $q$ with $-q$ in a yet to be specified
way.

With the last suggestion in mind we propose to `compactify' phase space
independent of $E$ in that $\T^\ast\!\rz$ is replaced by a torus, $\rz^2/\Gamma$,
where $\Gamma\cong\gz^2$ is a lattice. Any observable then has to be a 
function on $\rz^2$ that is periodic with respect to $\Gamma$. If the lattice
were chosen such that it acts on points in $\rz^2$ as 
$(q,p)\mapsto(q+na,p+mb)$, with some fixed parameters $a,b>0$ and
integers $n,m$, observables would have to be functions $f$ that are periodic
in $q$ and $p$. The symbol  $h_{\rm BK}(q,p)=qp$ clearly does not have this
property, so one would need to restrict it to a fundamental domain of $\Gamma$,
say, $0\leq q\leq a$, $0\leq p\leq b$, and extend it periodically to $\rz^2$.
The result would be a function on $\rz^2$ that is not continuous.

One can achieve an improvement by introducing new coordinates,
\begin{equation}
\label{eq:cantrafo}
 x=\frac{p-q}{\sqrt{2}}\quad\text{and}\quad\xi=\frac{p+q}{\sqrt{2}}.
\end{equation}
The map $(q,p)\mapsto(x,\xi)$ is symplectic (a canonical transformation),
and the symbol $h_{\rm BK}$ is transformed to
\begin{equation}
\label{eq:classhamilton}
 h(x,\xi)=\frac{\xi^2 -x^2}{2}.
\end{equation}
One then defines the lattice $\ell_x\gz\oplus\ell_\xi\gz$ acting
on $\rz^2$ as $(x,\xi)\mapsto(x+n\ell_x,\xi+m\ell_\xi)$, $n,m\in\gz$. Restricting
$h(x,\xi)$ to a fundamental domain $(-\ell_x/2,\ell_x/2)\times (-\ell_\xi/2,\ell_\xi/2)$
of the lattice and extending it periodically to $\rz^2$ then gives a continuous,
albeit not differentiable, function. Hence, the Fourier series for $h$ will have
improved convergence properties over the Fourier series for $h_{\rm BK}$.
As opposed to some alternative suggestions this procedure maintains the
symmetry in $x$ and $\xi$, at least when choosing $\ell_x=\ell_\xi$. A fundamental
domain of the lattice as well as an energy surface of $h$ are shown in
Figure~\ref{fig:energy_contours}(b).

Returning to operators defined in $L^2(\rz)$, a (Weyl) quantisation of the 
symbol \eqref{eq:classhamilton} produces the Hamiltonian operator for an 
inverted oscillator,
\begin{equation}
\label{eq:quantHamilton}
 H=-\frac{\hbar^2}{2}\frac{\ud^2}{\ud x^2}-\frac{1}{2}x^2,
\end{equation}
that one can define on the domain $C_0^\infty(\rz)$. 
\begin{prop}
\label{prop:specprop}
The operator $H$ has the following properties:
\begin{itemize}
\item[(i)] Defined on the domain $C_0^\infty(\rz)$, it is essentially self-adjoint.
(For simplicity we denote its self-adjoint extension also as $H$.)
\item[(ii)] It is unitarily equivalent to the operator $H_{\rm BK}$ defined in
\eqref{eq:BKoperator}.
\item[(iii)] $H$ has a purely absolutely continuous spectrum, 
$\sigma(H)=\sigma_{\rm ac}(H)=\rz$.
\end{itemize}
\end{prop}
\begin{proof}
The proof uses standard arguments:
\begin{itemize}
\item[(i)] This property follows from \cite[Satz 17.15] {Weidmann:2003}.
\item[(ii)] The symbol $h$ of the operator $H$ was obtained from the symbol
$h_{\rm BK}$ of the operator $H_{\rm BK}$ through the linear
symplectic map \eqref{eq:cantrafo}.  A combination of
\cite[Eq.~(2.1)]{Folland:1989} with \cite[Eq.~(4.23)]{Folland:1989}
therefore yields that $H$ and $H_{\rm BK}$ are unitarily
equivalent. More specifically, by referring to the Heisenberg
equations of motion for a harmonic oscillator one can confirm that,
indeed,
\begin{equation}
 \ue^{\ui\frac{\pi}{4\hbar}H_{\rm osc}}H_{\rm BK}\ue^{-\ui\frac{\pi}{4\hbar}H_{\rm osc}}
 =H,\quad\text{where}\quad 
 H_{\rm osc}=-\frac{\hbar^2}{2}\frac{\ud^2}{\ud x^2}+\frac{1}{2}x^2.
\end{equation}
\item[(iii)] In \cite[Proposition~6.2]{Per83} as well as in 
\cite[Satz 24.6]{Weidmann:2003} it is shown that the spectrum of 
$H_{BK}$ is $\rz$ and is purely absolutely continuous. 
By (ii) the same then holds for the operator $H$.
\end{itemize}
\end{proof}
We now turn to constructing a model operator by Weyl quantising the equivalent
of the symbol \eqref{eq:classhamilton} on the torus $\rz^2/\ell_x\gz\oplus\ell_\xi\gz$.
Hence, for all $n,m\in\gz$ we set
\begin{equation}
\label{eq:hamiltonper}
  h(x,\xi) = \frac{(\xi-m\ell_\xi)^2 - (x-n\ell_x)^2}{2}, \quad \text{if} \quad \begin{cases}
  \left(m-\tfrac{1}{2}\right) \ell_\xi\leq\xi < \left(m+\tfrac{1}{2}\right) \ell_\xi \\
  \left(n-\tfrac{1}{2}\right) \ell_x\leq x < \left(n+\tfrac{1}{2}\right) \ell_x  \end{cases}.
\end{equation}
This function has a representation as a Fourier series,
\begin{equation}
\label{eq:h_Fourier}
 h(x,\xi)=\sum_{n,m\in\gz}h_{mn}\,
 \ue^{2\pi\ui\bigl(\frac{m\xi}{\ell_\xi}-\frac{nx}{\ell_x}\bigr)},
\end{equation}
with coefficients
\begin{equation}
\label{eq:Fouriercoeff}
 h_{mn} = \begin{cases} \frac{\ell_\xi^2 -\ell_x^2}{24},&\ \text{if}\quad (m,n)=(0,0) \\
 \frac{1}{4\pi^2}\left(\ell_\xi^2\frac{(-1)^m}{m^2}(1-\delta_{m0})\delta_{n0}-  
 \ell_x^2\frac{(-1)^n}{n^2}(1-\delta_{n0})\delta_{m0} \right),&\ \text{if}\quad (m,n)\neq(0,0)
 \end{cases}.
\end{equation}
The Weyl quantisation on the torus was introduced in \cite{Hannay:1980}, for 
details see, e.g., \cite{Bouzouina:1996,Esposti:2003}. It assigns to the function 
\eqref{eq:hamiltonper} with Fourier coefficients \eqref{eq:Fouriercoeff}
an operator
\begin{equation}
\label{q34}
 \op_N(h) : =\sum_{m,n\in\gz}h_{mn}\,T^{m,n}
\end{equation}
acting in $\kz^N$. When $h$ is real valued and $\kz^N$ is equipped with the 
standard inner product, $\op_N(h)$ is self-adjoint. Here the unitary operators
\begin{equation}
\label{Weylnm}
 T^{m,n}=\ue^{\ui\pi\frac{nm}{N}}T^{m,0}T^{0,n}
\end{equation}
in $\kz^N$ are defined through the actions
\begin{equation}
\label{q32}
 \left( T^{m,0}\psi\right)_l:=\psi_{l+m} \quad\text{and}\quad
 \left( T^{0,n}\psi\right)_l:=\ue^{-\frac{2\pi\ui ln}{N}}\psi_l
\end{equation}
on $\psi=(\psi_l)\in\kz^N$,
where we use the convention that $\psi_{l+N}=\psi_l$. The operators $T^{m,0}$
and $T^{0,n}$ represent translations in the $x$- and $\xi$-directions,
respectively. The integer $N\in\nz$ is a semiclassical parameter, such that 
$N\to\infty$ is the semiclassical limit. It satisfies the relation
\begin{equation}
\label{eq:quant}
 2\pi\hbar N = \ell_\xi\ell_x .
\end{equation}
Therefore, the Weyl quantisation of the symbol \eqref{eq:hamiltonper} is
\begin{equation}\label{eq:hamitlonweyl}
\begin{split}
 \op_N(h)= \frac{\ell_\xi^2 -\ell_x^2}{24} \,T^{0,0}
  &+ \frac{\ell_\xi^2}{4\pi^2}\sum_{m=1}^\infty \frac{(-1)^m}{m^2}\,
        \bigl(T^{m,0}+T^{-m,0}\bigr) \\
  &- \frac{\ell_x^2}{4\pi^2}\sum_{n=1}^\infty \frac{(-1)^n}{n^2}\,\
         \bigl(T^{0,n}+T^{0,-n}\bigr),
\end{split}
\end{equation}
and acts on a vector $\psi=(\psi_l)\in\kz^N$ as
\begin{equation}
\begin{split}
 \bigl(\op_N(h)\psi\bigr)_l 
  &=  \frac{\ell_\xi^2 -\ell_x^2}{24} \,\psi_l
       + \frac{\ell_\xi^2}{4\pi^2}\sum_{m=1}^\infty \frac{(-1)^m}{m^2}\,
        \bigl(\psi_{l+m}+\psi_{l-m}\bigr) \\
  &\quad- \frac{\ell_x^2}{4\pi^2}\sum_{n=1}^\infty \frac{(-1)^n}{n^2}\,\
         \bigl(\ue^{-\frac{2\pi\ui ln}{N}}+\ue^{\frac{2\pi\ui ln}{N}}\bigr)\psi_l \\
   &=  \frac{\ell_\xi^2}{24} \,\psi_l
       + \frac{\ell_\xi^2}{4\pi^2}\sum_{m=1}^\infty \frac{(-1)^m}{m^2}\,
        \bigl(\psi_{l+m}+\psi_{l-m}\bigr) 
- \frac{1}{2}\left(\frac{l\ell_x}{N}\right)^2\psi_l  .     
\end{split}
\end{equation}
It can hence be seen as a difference operator on a lattice with $N$ points 
and periodic boundary conditions.

The spectrum $\sigma(\op_N(h))$ of $\op_N(h)$ consists of $N$ real eigenvalues 
$E_n$.
\begin{lemma}
\label{specsym}
When $\ell_x=\ell_\xi$ the spectrum of $\op_N(h)$ is symmetric about zero, i.e.,
$E_n\in\sigma(\op_N(h))$ implies that $-E_n\in\sigma(\op_N(h))$.
\end{lemma}
We prove this lemma in Appendix~\ref{appA2}.

Below we wish to study the distribution of the eigenvalues of $\op_N(h)$ in the 
semiclassical limit $N\to\infty$. We intend our approach to provide an approximation
to the operator \eqref{eq:quantHamilton} and, therefore, keep the value of $\hbar$
fixed. Due to the relation \eqref{eq:quant}, the semiclassical limit $N\to\infty$ can
then be achieved by sending $\ell_\xi$ and/or $\ell_x$ to infinity.

\section{Continuum and infinite volume limit}
\label{sec:limits}
Before studying spectral properties of the operator $\op_N(h)$ in
various combinations of the limits $\ell_\xi\to\infty$ and $\ell_x\to\infty$,
we want to explore their interpretation in some more detail. It is obvious
from the set-up that $\ell_x$ measures the extension of the model in
`configuration space', whereas $\ell_\xi$ is the corresponding measure
in `momentum space'. One may then view vectors $\psi=(\psi_n)\in\kz^N$
as functions evaluated at the $N$ equidistant points
\begin{equation}
 x_n = \frac{n\ell_x}{N}\in \left[-\frac{\ell_x}{2},\frac{\ell_x}{2}\right),
\end{equation}
where $-\frac{N}{2}\leq n<\frac{N}{2}$. The distance between neighbouring 
points can be expressed as
\begin{equation}
\label{xn}
\frac{\ell_x}{N}=\frac{2\pi\hbar}{\ell_\xi}
\end{equation}
by making use of the relation \eqref{eq:quant}. As we keep $\hbar$ fixed,
taking $\ell_\xi\to\infty$ corresponds to a continuum limit, in which difference
operators approximate differential operators. In contrast, taking 
$\ell_x\to\infty$ corresponds to an infinite-volume limit. Both limits are 
semiclassical in the sense of Weyl quantisation on a torus as they imply 
$N\to\infty$.

When a symbol only depends on either $x$ or $\xi$, the spectrum of its
quantisation can be determined explicitly. As a simplification of
\eqref{eq:hamiltonper} we therefore choose 
\begin{equation}
\label{eq:symbol_xi^2}
  a(\xi) = (\xi-m\ell_\xi)^2 ,\quad \text{if} \quad 
  \left( m - \tfrac{1}{2} \right) \ell_\xi \leq \xi 
  < \left( m + \tfrac{1}{2} \right) \ell_\xi \, ,\quad m\in\gz,
\end{equation}
with Fourier expansion
\begin{equation}
\label{eq:Fourier_xi^2}
  a(\xi) = \frac{\ell_\xi^2}{12} + \frac{\ell_\xi^2}{\pi^2} 
             \sum_{n=1}^\infty \frac{(-1)^n}{n^2} 
               \cos\left( \frac{2\pi}{\ell_\xi}n\xi\right).
\end{equation}
Its quantisation is
\begin{equation}
\label{eq:opNa}
  \op_N(a) = \frac{\ell_\xi^2}{12} + \frac{\ell_\xi^2}{2\pi^2} 
             \sum_{n=1}^\infty \frac{(-1)^n}{n^2} \,
               \bigl(T^{n,0}+T^{-n,0}\bigr) .
\end{equation}
We define plane waves $\psi^{(\nu)}=(\psi^{(\nu)}_m)\in\kz^N$ as
\begin{equation}
\label{eq:planewave}
  \psi^{(\nu)}_m 
  = \ue^{\frac{\ui}{\hbar}\frac{\nu\ell_\xi}{N}\frac{m\ell_x}{N}}
  = \ue^{2\pi\ui\frac{\nu m}{N}}.
\end{equation}
Choosing $\nu\in\mathbb{Z}$ with $-\frac{N}{2} \leq \nu < \frac{N}{2}$
ensures that the $N$ momenta satisfy
\begin{equation}
 \frac{\nu\ell_\xi}{N}\in\left[-\frac{\ell_\xi}{2},\frac{\ell_\xi}{2}\right).
\end{equation}
Plane waves can easily be seen to be eigenfunctions of the operators 
$T^{n,0}$,
\begin{equation}
  (T^{n,0}\psi^{(\nu)})_m = \ue^{2\pi\ui\frac{\nu n}{N}} \psi^{(\nu)}_m , 
\end{equation}
and consequently 
\begin{equation}
  (\op_N(a) \psi^{(\nu)})_m
  = \left( \frac{\ell_\xi^2}{12} + \frac{\ell_\xi^2}{\pi^2} 
           \sum_{n=1}^\infty \frac{(-1)^n}{n^2} \,
           \cos\left( 2\pi\frac{\nu n}{N} \right) \right) \psi^{(\nu)}_m
  = \left( \frac{\nu\ell_\xi}{N} \right)^2 \psi^{(\nu)}_m .
\end{equation}
It is convenient to express the eigenvalues 
$E^{N}_\nu = \left( \frac{\nu\ell_\xi}{N} \right)^2$ in several ways, 
\begin{subequations}
\label{eq:eigenvalues_xi^2}
\begin{align}
  \label{eq:eigenvalues_xi^2-lxi_fix}
  E^{N}_\nu 
  &= \ell_\xi^2 \frac{\nu^2}{N^2} \\
  \label{eq:eigenvalues_xi^2-lx_fix}
  &= \left(\frac{2\pi\hbar}{\ell_x}\right)^2 \nu^2 \\
  \label{eq:eigenvalues_xi^2_nixfix}
  &= 2\pi\hbar \frac{\ell_\xi}{\ell_x} \frac{\nu^2}{N} \, .
\end{align}
\end{subequations}
We note in passing that
\begin{equation}
\label{eq:xi^2-spectral_support}
  \inf\sigma(\op_N(a)) = 0 \qquad\text{and}\qquad
  \sup\sigma(\op_N(a)) = \left\{ \begin{matrix} 
  \frac{\ell_\xi^2}{4} & \text{for $N$ even} \\[.5ex]
  \frac{\ell_\xi^2}{4} \left(\frac{N-1}{N}\right)^2 & 
    \text{for $N$ odd} \end{matrix} \right. ,
\end{equation}
and that the nearest-neighbour separation of eigenvalues 
(for non-negative $\nu$) is
\begin{subequations}
\label{eq:differences_xi^2}
\begin{align}
  \label{eq:differences_xi^2-lxi_fix}
  s^{N}_\nu 
  = E^{N}_{\nu+1} - E^{N}_\nu 
  &= \ell_\xi^2 \, \frac{2\nu+1}{N^2}\\
  \label{eq:differences_xi^2-lx_fix} 
  &= \left(\frac{2\pi\hbar}{\ell_x}\right)^2 (2\nu+1)\\
  \label{eq:differences_xi^2-nixfix} 
  &= 2\pi\hbar \frac{\ell_\xi}{\ell_x} \frac{2\nu+1}{N}.
\end{align}
\end{subequations}
We now discuss the behaviour of the spectrum in the continuum-
and/or infinite-volume limit:
\begin{enumerate}
\item {\it Continuum limit in finite volume}, i.e., $\ell_x$ fixed, $\ell_\xi\to\infty$: 
From \eqref{eq:eigenvalues_xi^2-lx_fix} one sees that the eigenvalues,
even before taking the limit, are the same as eigenvalues of the operator
$-\hbar^2\frac{\ud^2}{\ud x^2}$ on an interval $[-\ell_x/2,\ell_x/2]$ with
periodic boundary conditions. The only effect of the limit $N\to\infty$ hence
is a growing number of eigenvalues, eventually covering the entire spectrum
of the differential operator. The separation of neighbouring eigenvalues 
\eqref{eq:differences_xi^2-lx_fix} is unaffected. 
\item {\it Infinite volume limit on the lattice 
$\frac{2\pi\hbar}{\ell_\xi}\mathbb{Z}$}, i.e., $\ell_\xi$ fixed, $\ell_x\to\infty$: 
In this limit the expression on the right-hand side of \eqref{eq:opNa} for the
operator $\op_N(a)$ is unaffected by the limit; $\op_N(a)$ remains 
a difference operator, albeit on a growing lattice that tends to an infinite
lattice in the limit. From \eqref{eq:xi^2-spectral_support} we see that 
the spectrum remains in the finite interval $[0,\ell_\xi^2/4]$, whereas
\eqref{eq:differences_xi^2-lxi_fix} shows that the eigenvalues become 
dense everywhere in this interval.
\item {\it Continuum limit in infinite volume}, i.e., $\ell_x\to\infty$ and 
$\ell_\xi\to\infty$: Combining the two previous cases one might expect
a convergence (in a suitable sense) of the difference operator $\op_N(a)$ 
to the differential operator $-\hbar^2\frac{\ud^2}{\ud x^2}$ acting in $L^2(\rz)$.
The spectrum is expected to approach a `continuous' spectrum 
$[0,\infty)$. Indeed, \eqref{eq:xi^2-spectral_support} implies that
$\sup\sigma(\op_N(a))\to\infty$ in this case, but whether the
eigenvalues become dense depends on how fast $\ell_x$ and $\ell_\xi$
go to $\infty$, and the answer can be different in different ranges. 
To illustrate this we set 
\begin{equation}
  \ell_\xi = A N^\alpha \, , \quad \ell_x = B N^{1-\alpha}
\end{equation}
with $\alpha\in(0,1)$ and constants $A,B$ satisfying $AB=2\pi\hbar$
in order to comply with \eqref{eq:quant}. Then according to
\eqref{eq:eigenvalues_xi^2_nixfix} and \eqref{eq:differences_xi^2-nixfix},
\begin{equation}
 E^{N}_\nu  = A^2 N^{2\alpha-2}\,\nu^2 \quad\text{and}\quad
 s^{N}_\nu = A^2 N^{2\alpha-2}\,(2\nu+1).
\end{equation}
Thus, at a fixed range in the spectrum, i.e.\ when $\nu$ grows like $N^{1-\alpha}$, 
the eigenvalues become dense for every $\alpha>0$. Near the supremum 
of the spectrum, i.e., when $\nu$ grows proportional to $N$, the eigenvalues only 
become dense if $\alpha<\frac{1}{2}$, i.e., if $\ell_x$ grows faster than $\ell_\xi$.
\end{enumerate}
When $\ell_x=\ell_\xi$, the operators $\op_N(a)$ and $\op_N(b)$, where 
$b(x)=x^2$, are unitarily equivalent via the discrete Fourier transform, see 
\eqref{eq:aFb}. Therefore, the operators have the same spectrum and the above
discussion applies to $\op_N(b)$ too. However, even when $\ell_x\neq\ell_\xi$ 
can one carry over the above analysis to $\sigma(\op_N(b))$; one only needs 
to swap $\ell_\xi$ and $\ell_x$.

For the operator $\op_N(h)$ we expect a qualitatively similar behaviour.
We first note that since
\begin{equation}
 \op_N(h) = \frac{1}{2}\bigl( \op_N(a) - \op_N(b) \bigr),
\end{equation}
and both $\op_N(a)$ and $\op_N(b)$ are positive operators, the min-max 
principle allows us to bound the spectrum of $\op_N(h)$ from above and 
below in terms of the upper bounds $E_{\rm a,max}=\ell_\xi^2/4$ and 
$E_{\rm b,max}=\ell_x^2/4$ for $\op_N(a)$ and $\op_N(b)$, respectively,
\begin{equation}
\label{eq:spechbound}
 -\frac{\ell_x^2}{8}\leq \op_N(h) \leq \frac{\ell_\xi^2}{8}.
\end{equation}
In the various cases of sending $\ell_\xi$ and/or $\ell_x$ to infinity we make 
the following observations:
\begin{enumerate}
\item {\it Continuum limit in finite volume}, i.e., $\ell_x$ fixed, $\ell_\xi\to\infty$: 
The explicit form \eqref{eq:matrix_elements_xi^2-x^2} of the matrix elements
of $\op_N(h)$ show that in this limit $\op_N(a)$ dominates, so that $\op_N(b)$
can be seen as a perturbation. Hence, as $\ell_\xi\to\infty$ the eigenvalues of
$\op_N(h)$ will asymptotically be given by (one half of) the right-hand side of
\eqref{eq:eigenvalues_xi^2} ($\nu$ fixed). This will describe large eigenvalues;
at the bottom of the spectrum the contribution of $\op_N(b)$ will have an effect.
As can be seen from \eqref{eq:spechbound}, in the limit the spectrum will
be contained in $[-\ell_x^2/8,\infty)$. We expect the limiting operator to be
\begin{equation}
 -\frac{\hbar^2}{2}\frac{\ud^2}{\ud x^2} - \frac{1}{2}x^2
\end{equation}
on the interval $[-\ell_x/2,\ell_x/2]$ with periodic boundary conditions.
\item {\it Infinite volume limit on the lattice 
$\frac{2\pi\hbar}{\ell_\xi}\mathbb{Z}$}, i.e., $\ell_\xi$ fixed, $\ell_x\to\infty$: 
This limit describes the opposite case to the previous one, in that $\op_N(a)$
will be a perturbation of $-\op_N(b)$. Asymptotically, the eigenvalues will be
given by $-\tfrac{1}{2}$ times the right-hand side of \eqref{eq:eigenvalues_xi^2}
with $\ell_x$ instead of $\ell_\xi$. The limiting spectrum will be contained in
$(-\infty,\ell_\xi^2/8]$.
\item {\it Continuum limit in infinite volume}, specifically, 
$\ell:=\ell_x=\ell_\xi\to\infty$: From \eqref{eq:matrix_elements_xi^2-x^2} one
obtains that all matrix elements of $\op_N(h)$ contain a factor $\ell^2$, hence
the eigenvalues will have the same factor and otherwise be independent of
$\ell$. The only conclusion one can draw from the bounds \eqref{eq:spechbound}
is that the limiting spectrum will be contained in $\rz$. We expect the limiting
operator to be $H$ \eqref{eq:quantHamilton}, whose spectrum is described
in Proposition~\ref{prop:specprop}.
\end{enumerate}
The semiclassical discussion in the following section will confirm this picture.

\section{Semiclassics}
\label{sec:semiclassics}
For a semiclassical analysis of the spectrum of $\op_N(h)$ we need to determine
some classical quantities, including the energy surface and periodic orbits of the 
Hamiltonian flow, as well as their actions and periods. We first note that 
\begin{equation}
 -\frac{\ell_x^2}{8}\leq E\leq\frac{\ell_\xi^2}{8},
\end{equation}
which is the classical equivalent to \eqref{eq:spechbound}. For these values 
of $E$ the energy surface is
\begin{equation}
 \Sigma_E = \left\{(x,\xi)\in\left(-\frac{\ell_x}{2},\frac{\ell_x}{2}\right)\times
 \left(-\frac{\ell_\xi}{2},\frac{\ell_\xi}{2}\right);\ \frac{\xi^2-x^2}{2}=E  \right\}.
\end{equation}
When $\ell_\xi=\ell_x$, the energy surface is always connected. Otherwise,
depending on the values for $\ell_\xi,\ell_x,E$ it is either connected or consists 
of two connected components. The latter case occurs when either
$\ell_\xi>\ell_x$ and 
\begin{equation}
\label{eq:posenrange}
 0<E<\frac{\ell_\xi^2-\ell_x^2}{8},
\end{equation}
or when $\ell_x>\ell_\xi$ and
\begin{equation}
 \frac{\ell_\xi^2-\ell_x^2}{8}<E<0.
\end{equation}
In all other cases $\Sigma_E$ is connected. Likewise, when $\Sigma _E$ is 
connected, there is only one primitive periodic orbit $p$ of the Hamiltonian flow.
Its period $t_p$ then is the volume (in Liouville measure) of $\Sigma_E$.
When $\Sigma_E$ has two components, however, there are two primitive 
periodic orbits $p$ and $p^{-1}$, where the latter is the time reversal of the former.
In this case the volume of $\Sigma_E$ is the sum of $t_p$ and $t_{p^{-1}}$,
i.e., $2t_p$.

When one keeps $\ell_x$ fixed and increases $\ell_\xi$, only in the small 
range $\ell_\xi^2-\ell_x^2<8E<\ell_\xi^2$ will $\Sigma_E$ be connected. In the
large energy range \eqref{eq:posenrange} $\Sigma_E$ has two components.
As $E$ grows, these components approach the two components of the energy 
surface 
\begin{equation}
 \Sigma_E^0 :=\left\{ (x,\xi);\ \xi=\pm\sqrt{2E}\right\}
\end{equation}
of $\xi^2/2$. A Bohr-Sommerfeld quantisation of the energy surfaces 
$\Sigma_E$ hence demonstrates that semiclassical approximations of large 
eigenvalues of $\op_N(h)$ in the range \eqref{eq:posenrange} approach 
semiclassical approximations of the eigenvalues of the operator
$-\frac{\hbar^2}{2}\tfrac{\ud^2}{\ud x^2}$ on the interval $[-\ell_x/2,\ell_x/2]$ 
with periodic boundary conditions. (The latter are identical with the exact 
eigenvalues \eqref{eq:eigenvalues_xi^2-lx_fix}.) Near the bottom of the spectrum 
one would approximate eigenvalues with Bohr-Sommerfeld quantisations of 
surfaces $\Sigma_E$ at small, positive values of $E$, which deviate essentially 
from the corresponding energy surfaces $\Sigma_E^0$. There will also be a 
limited range of negative energies with connected $\Sigma_E$ whose
quantisations may approximate negative eigenvalues. Thus, the semiclassical
picture conforms with our discussion in Section~\ref{sec:limits}.

When $\ell_x$ grows while $\ell_\xi$ is kept fixed, the above discussion applies
after having swapped $x$ and $\xi$ as well as positive $E$ and negative $E$.
In the case $\ell_x=\ell_\xi$ all energy surfaces are connected and there is a
complete symmetry in $x$ and $\xi$. After Bohr-Sommerfeld quantisation this 
leads to a semiclassical equivalent of Lemma~\ref{specsym}. 

We now calculate the action $S_p$ and period $t_p$ of the single primitive 
periodic orbit $p$ in the case of a connected energy surface at positive 
energy $E$,
\begin{equation}
 S_p = \oint_p \xi\,\ud x \qquad\text{and}\qquad t_p = \frac{\ud S_p}{\ud E}.
\end{equation}
At negative energy one has to swap $x$ and $\xi$, but otherwise gets the
same expressions. In the fundamental domain the primitive orbit $p$ consists 
of the two sections of the hyperbola $\xi^2-x^2=2E$ between the points 
$(x_-,\ell_\xi/2)$ and $(x_+,\ell_\xi/2)$, and $(x_+,-\ell_\xi/2)$ and $(x_-,-\ell_\xi/2)$, 
respectively, as shown in Figure~\ref{fig:energy_contours}(b). Here we have defined 
$x_\pm=\pm\sqrt{\ell_\xi^2/4 +2E}$. Thus, for the action we find
\begin{equation}
\label{eq:paction}
\begin{split}
 S_p &= 2\int_{x_-}^{x_+}\sqrt{x^2+2E}\,\ud x +\ell_\xi\bigl( x_+ -x_-\bigr)\\
         &= 4E\arsinh\sqrt{\frac{\ell_\xi^2}{8E}-1} - \frac{\ell_\xi}{2}
               \sqrt{\ell_\xi^2 -8E}.
\end{split}
\end{equation}
The derivative of this result yields the period,
\begin{equation}
\label{eq:pperiod}
 t_p = 4\arsinh\sqrt{\frac{\ell_\xi^2}{8E}-1}.
\end{equation}
With these classical quantities available, we can use them in the trace formula
that was proven in \cite{BEK16}. Choosing a test function $\rho\in C^\infty(\rz)$
with Fourier transform $\hat\rho\in C_0^\infty(\rz)$, denoting the eigenvalues
of $\op_N(h)$ as $E_n$, the Maslov index of the primitive periodic orbit $p$ 
as $\sigma_p$ \cite{BEK16}, and with the notation as above, this trace formula reads
\begin{equation}
 \sum_n\rho\left(\frac{E_n-E}{\hbar}\right) = \sum_{k\in\gz}\frac{t_p\,\hat\rho(kt_p)}{2\pi}
 \,\ue^{\frac{\ui}{\hbar}kS_p-\ui\frac{\pi}{2}k\sigma_p}+O(\hbar).
\end{equation}
Using a suitable Tauberian theorem \cite[Theorem 6.3]{BruPauUri95}, this expression 
allows us to determine the leading semiclassical behaviour of the local eigenvalue 
counting function to be
\begin{equation}
 N(E;r) := \#\{n;\ |E_n -E|\leq r\hbar\}
                \sim \frac{r}{\pi}\,t_p \quad\text{as}\quad\hbar\to 0.
\end{equation}
This means that the leading asymptotic behaviour of the spectral density
is given by
\begin{equation}
\label{eq:semicldense}
 d(E) \sim \frac{t_p}{2\pi\hbar} 
         = \frac{2}{\pi\hbar}\arsinh\sqrt{\frac{\ell_\xi^2}{8E}-1}.
\end{equation}
We now propose to link our length-parameters $\ell_\xi,\ell_x$ to the 
truncation \eqref{eq:pstrunc} of the phase space introduced by Berry and Keating. 
Their parameters $L_q,L_p$ measure the closest distance of the truncated 
hyperbola $qp=E$ to the coordinate axes. In our coordinates this would be the 
closest distance of the hyperbola $\xi^2-x^2=2E$ to the asymptotes $\xi=\pm x$.
The truncation at $E/L_p$ on the $q$-axis in \cite{BerKea99a} then corresponds 
to the distance 
\begin{equation}
 \sqrt{2}\frac{\ell_\xi}{2}-L_p
\end{equation}
on the asymptote $\xi=x$, see Figure~\ref{fig:energy_contours}(c) for an illustration. 
Equating these quantities and using the value $L_p=\sqrt{2\pi}$ (in units where $\hbar=1$) 
as in \cite{BerKea99a} gives
\begin{equation}
\label{rel1}
 \ell_\xi = \frac{E+2\pi}{\sqrt{\pi}},
\end{equation}
hence linking the torus-length scale to the energy. With this identification 
the action \eqref{eq:paction} and the period \eqref{eq:pperiod} become
\begin{equation}
\begin{split}
 S(E) &= 4E\arsinh\left(\frac{E-2\pi}{\sqrt{8\pi E}}\right) -\frac{E^2}{2\pi} +2 \\
         &= 2E\log\frac{E}{2\pi} -\frac{E^2}{2\pi} +2
\end{split}
\end{equation}
and 
\begin{equation}
\label{t_p}
t_p = 4E\arsinh\left(\frac{E-2\pi}{\sqrt{8\pi E}}\right)
      =2\log\left(\frac{E}{2\pi}\right).
\end{equation}
Using this in the expression on the right-hand side of \eqref{eq:semicldense}
then leads us to expect a logarithmic mean spectral density
\begin{equation}
\label{eq:sc_mean_density}
 \bar d(E) = \frac{1}{\pi}\,\log\left(\frac{E}{2\pi}\right).
\end{equation}
We note that this is twice the value which Berry and Keating \cite{BerKea99a}
found in their model. The factor of two arises from the fact that in
\cite{BerKea99a} only one branch of the hyperbola was considered, whereas
we have taken both branches into account, see Figure~\ref{fig:energy_contours}.

\section{Numerical results}
\begin{figure}[t]
\begin{center}
\includegraphics[width=.5\textwidth]{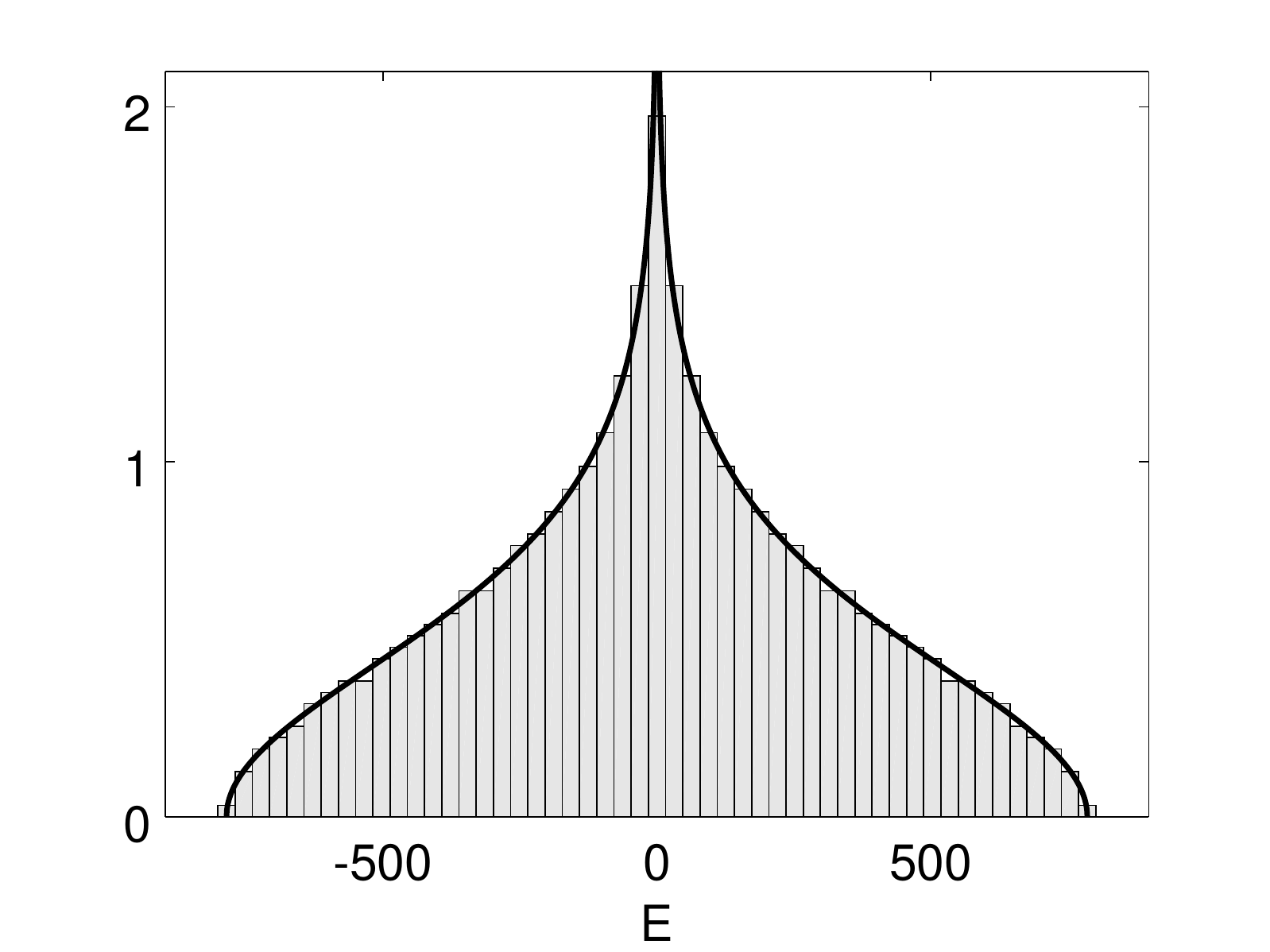}
\caption{Histogram (grey) of the eigenvalues of $\op_N(h)$ for
  $N=1000$, $\hbar=1$ and $\ell_x=\ell_\xi=\sqrt{2\pi N}$ compared
  to the semiclassical estimate (thick black line) for the density
  \eqref{eq:semicldense}.}
\label{fig:pickelhaube}
\end{center}
\end{figure}
\begin{figure}[t]
\begin{center}
\includegraphics[width=.5\textwidth]{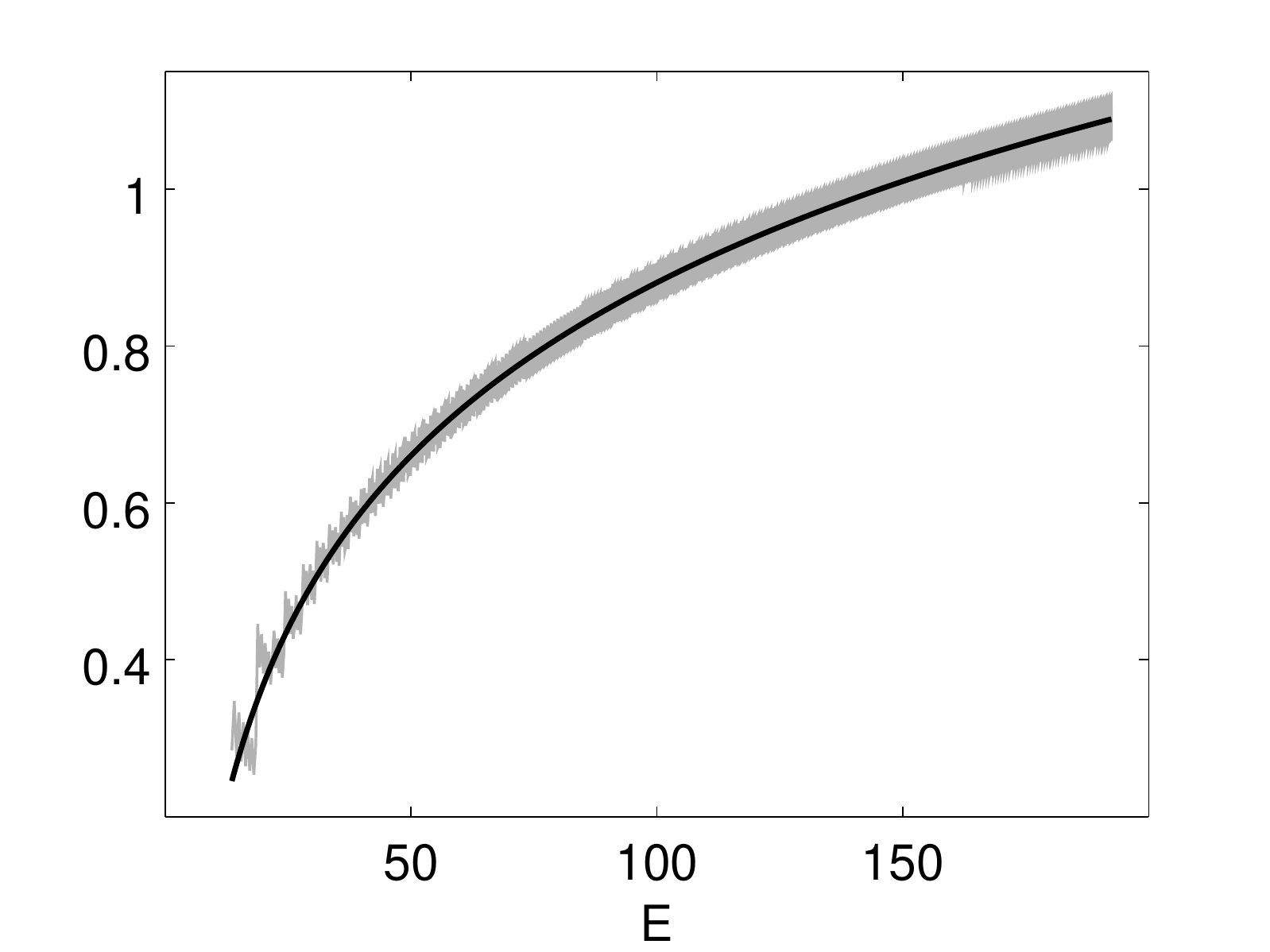}%
\includegraphics[width=.5\textwidth]{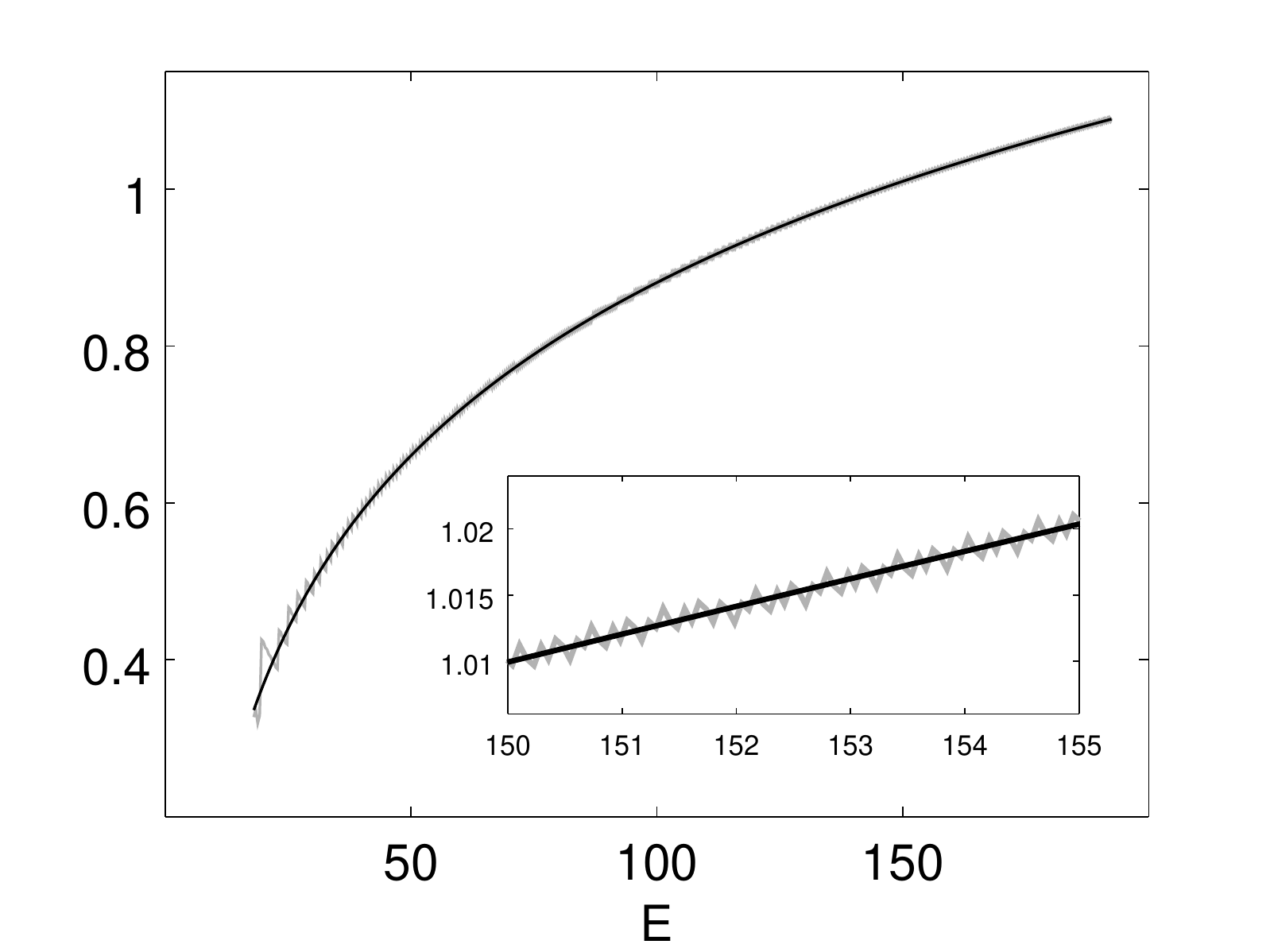}
\caption{Numerical estimate $d^K$ (grey) of the mean spectral density
  compared to the semiclassical asymptotics $\bar{d}$ (black), see
  Eq.~\eqref{eq:sc_mean_density}. The left panel is for $K=2$; one
  observes fluctuations of $d^K$ around $\bar{d}$. For $K=3$ (right
  panel) the lines are hardly distinguishable at higher energies; small
  rapid fluctuations are visible upon magnification (inset).}
\label{fig:dquer}
\end{center}
\end{figure}
In order to illustrate the findings of Section~\ref{sec:semiclassics} we
have numerically diagonalised the Hamiltonian $\op_N(h)$ in the form
\begin{equation}
\label{eq:ophfinite}
 \op_N(h) = \frac{\pi}{2N}\sum_{m=1}^{N-1}
  \frac{(-1)^m}{\sin^2\!\left(\pi\frac{m}{N}\right)}
  \left(T^{m,0}-T^{0,m}\right) \left\{ \begin{matrix} 
       1 & \text{if $N$ even} \\[.5ex] 
       \cos\left(\pi\frac{m}{N}\right) & \text{if $N$ odd} \end{matrix}
     \right. \ . 
\end{equation}
Here, using Lemma~\ref{lemma:finite_symbols} and Eq.~\eqref{eq:gm0} in
Appendix~\ref{appB}, the infinite sum in \eqref{eq:hamitlonweyl} has
been converted into a finite sum.  Furthermore, we have chosen
$\hbar=1$ and $\ell_x=\ell_\xi=\sqrt{2\pi N}$.  This choice ensures
that the condition \eqref{eq:quant} is satisfied, as well as that the
continuum and the infinite-volume limits are performed
simultaneously. The matrix elements of this operator, that we have
used for the numerical diagonalisation, are given in
Eq.~\eqref{eq:matrix_elements_xi^2-x^2}.

We first test the accuracy of the semiclassical arguments leading to
the estimate \eqref{eq:semicldense} for the spectral density. To this
end we show a histogram of the eigenvalues of $\op_N(h)$ for $N=1000$
in Figure~\ref{fig:pickelhaube}. The overall behaviour is well described 
by Eq.~\eqref{eq:semicldense}, which is plotted for comparison.

Now we want to compare the logarithmic mean density
\eqref{eq:sc_mean_density} to the numerical data, where we have
to satisfy the condition~\eqref{rel1}. However, since we have already
fixed $\ell_\xi$ in terms of $N$, this relation now reads
\begin{equation}
  \sqrt{2\pi N}=\frac{E+2\pi}{\sqrt{\pi}} \, .
\end{equation}
It can thus be implemented in the following way. For each value of $N$
-- recall that $N$ labels matrices of different size, and thus
different spectra -- we have to determine the spectral density at the
energy
\begin{equation}
  E =\pi\sqrt{2N}-2\pi \, .
\end{equation}
The latter we do as follows. We first determine the $K$ eigenvalues
which are closest to the given energy $E$. Of these we determine the
minimal and the maximal value $E^K_\mathrm{min}$ and
$E^K_\mathrm{max}$, respectively. We then estimate the local spectral
density as
\begin{equation}
  d^K(E) = \frac{K-1}{E^K_\mathrm{max}-E^K_\mathrm{min}} \, .
\end{equation}
In Figure~\ref{fig:dquer} we compare the spectral densities $d^K$
obtained from spectra with $10K \leq N \leq 2000$ to the logarithmic
semiclassical estimate \eqref{eq:sc_mean_density} and observe a good
agreement.

\section{Conclusions}
\label{sec:conclusions}
Our goal was to introduce a self-adjoint quantum Hamiltonian 
as a modification of the Berry-Keating operator 
\cite{BerKea99a,BerKea99b} that has, to leading order, an eigenvalue 
count resembling the logarithmic law of the Riemann zeros. In order 
to achieve this we followed the suggestion of a phase-space truncation
made in \cite{BerKea99a,BerKea99b}. In contrast to previous approaches 
\cite{SieRod11,BerKea11} we did not modify the classical Hamiltonian 
to achieve compact energy surfaces, but we replaced the non-compact 
phase space $\T^\ast\!\rz$ by a torus. We then demonstrated that Weyl 
quantisation on a torus provides a powerful tool to construct operators 
and to analyse their properties. This setting enabled us to use
the three available parameters, the two length scales and the semiclassical
parameter, to perform various limits, including a continuum and an 
infinite-volume limit.

Representing the quantum Hamiltonian in Weyl quantisation offers the 
opportunity to use semiclassical methods in order to study spectral properties.
Based on the Gutzwiller trace formula for our setting \cite{BEK16}, we
derived a semiclassical expression for the spectral density which, however, 
is not of the desired logarithmic form. It rather expresses the spectral 
density in the vicinity of a fixed energy $E$, with fixed length parameters
and in a semiclassical approximation. However, it does fit a numerically
computed spectral density very well. As in the (semi-)classical
calculation provided in \cite{BerKea99a,BerKea99b}, we had to link the
torus length scales to the energy; this can be done in a similar way as 
in \cite{BerKea99a,BerKea99b}, producing the same leading term.

Our model provides a clean and direct implementation of the original 
proposal made by Berry and Keating, avoiding ad hoc modifications as 
well as non-local terms in the quantum Hamiltonian as they were suggested 
in \cite{SieRod11,BerKea11}. The original model was based on the work 
of Connes \cite{Con96,Con99}, and central to it was the idea to provide 
a quantisation of the classical Hamiltonian $H(q,p)=qp$, with the expectation 
that the hyperbolic nature of the classical dynamics would provide 
the necessary instability to produce a spectrum that would be able to 
mimic the `spectrum' of Riemann zeros, whose correlations on the scale 
of the mean level spacing can be modelled by random-matrix theory 
(see, e.g., \cite{BerKea99b} and references therein). In one degree of 
freedom, however, classical dynamics are integrable and, therefore, one 
cannot truly expect such a mimicking to happen. This is reflected by
the fact that we had to use a family of operators, parametrised by 
the semiclassical parameter $N$, and to link the size of the torus to 
the energy $E$. We then counted eigenvalues around $E$ in an $N$-dependent 
way, eventually leading to the desired logarithmic spectral density. 
In that way our analysis clearly demonstrates the limitations of 
modelling the Riemann zeros with the eigenvalues of an operator that 
is a quantisation (of an established nature) of $H(q,p)=qp$.

\vspace*{0.5cm}
\subsection*{Acknowledgements}
This research was supported through the programme {\it Research in
Pairs} by the Mathematisches Forschungsinstitut Oberwolfach in 2016.
A first stage of the work was supported by Grant no.\ 41534 of the 
London Mathematical Society.

\appendix 
\section{Proof of Lemma~\ref{specsym}}
\label{appA2}
Along the symbol $a(\xi)$ defined in \eqref{eq:symbol_xi^2}, whose Fourier 
series is \eqref{eq:Fourier_xi^2}, we define
\begin{equation}
\label{eq:symbol_x^2}
  b(x) = (x-n\ell_x)^2 ,\quad \text{if} \quad 
  \left( n - \tfrac{1}{2} \right) \ell_x \leq x 
  < \left( n + \tfrac{1}{2} \right) \ell_x \, ,\quad n\in\gz,
\end{equation}
with Fourier expansion
\begin{equation}
\label{eq:Fourier_x^2}
  b(x) = \frac{\ell_x^2}{12} + \frac{\ell_x^2}{\pi^2} 
             \sum_{n=1}^\infty \frac{(-1)^n}{n^2} 
               \cos\left( \frac{2\pi}{\ell_x}nx\right),
\end{equation}
and quantisation
\begin{equation}
\label{eq:opNb}
  \op_N(b) = \frac{\ell_x^2}{12} + \frac{\ell_x^2}{2\pi^2} 
             \sum_{n=1}^\infty \frac{(-1)^n}{n^2} \,
               \bigl(T^{0,n}+T^{0,-n}\bigr) .
\end{equation}
Hence, $h(x,\xi)=\tfrac{1}{2}(a(\xi)-b(x))$.

We also introduce the discrete Fourier transform $F$, which is a unitary 
operator on $\kz^N$ defined as
\begin{equation}
 (F\psi)_k := \frac{1}{\sqrt{N}}\sum_{l=0}^{N-1}
                    \ue^{-2\pi\ui\frac{kl}{N}}\,\psi_l\, ,
\end{equation}
where we recall the periodicity $\psi_{l+N}=\psi_l$. A short calculation 
then shows that
\begin{equation}
 F^{-1}T^{m,0}F = T^{0,m},\qquad\text{hence}\qquad FT^{0,m}F^{-1} = T^{m,0}.
\end{equation}
When $\ell_x=\ell_\xi$, a comparison of \eqref{eq:opNa} and \eqref{eq:opNb} 
hence yields that
\begin{equation}
\label{eq:aFb}
 F^{-1}\op_N(a)F = \op_N(b).
\end{equation}
Moreover, since $F^2$ is a parity operator, $(F^2\psi)_l = \psi_{-l}$, it follows that
\begin{equation}
 F^2\op_N(b)F^{-2} = \op_N(b),\quad\text{hence}\quad 
 \op_N(a) = F\op_N(b)F^{-1} = F^{-1}\op_N(b)F.
\end{equation}
Thus,
\begin{equation}
\begin{split}
 F^{-1}2\op_N(h)F &= F^{-1}\op_N(a)F - F^{-1}\op_N(b)F \\
                              &= \op_N(b) - \op_N(a) = -2\op_N(h).
\end{split}
\end{equation}
As $F^{-1}\op_N(h)F$ and $\op_N(h)$ have the same spectra, the lemma
is proven. 
\qed
 
\section{Matrix elements}
\label{appB}
In this appendix we calculate the matrix elements $\op_N(h)_{k,l}$ of the 
operator with Weyl symbol \eqref{eq:hamiltonper}. 

The matrix elements $A_{k,l}$ of an operator $A$ are defined by
\begin{equation}
  (A \psi)_k = \sum_{-\frac{N}{2}\leq k<\frac{N}{2}} A_{k,l} \, \psi_l \, .
\end{equation}
In order to calculate the matrix elements of a Weyl operator
\eqref{q34} we need the matrix elements of the unitary phase-space
translation operators $T^{m,n}$. From Eqs.~\eqref{Weylnm} and
\eqref{q32} one concludes that
\begin{equation}
 (T^{m,n}\psi)_k 
  = \sum_{-\frac{N}{2}\leq k<\frac{N}{2}} \ue^{-\ui\pi\frac{nm}{N}} \, 
    \ue^{-2\pi\ui\frac{kn}{N}} \, \delta_{k+m,l} \, \psi_l \, ,
\end{equation}
from which the matrix elements are read off as
\begin{equation}
\label{eq:translatmatrix}
   (T^{m,n})_{k,l} =  \ue^{-\ui\pi\frac{nm}{N}} \, 
    \ue^{-2\pi\ui\frac{kn}{N}} \, \delta_{k+m,l}\, .
\end{equation}
Here it is understood that our convention $\psi_{l+N}=\psi_l$ implies
that all indices are to be taken modulo $N$. 

Since phase space is a torus, the resulting periodicity of translations 
imply that for fixed $N$ we can always view a Weyl operator as the
quantisation of a symbol that is a finite trigonometric polynomial
instead of the in general infinite series \eqref{q34}.
\begin{lemma}
\label{lemma:finite_symbols}
Let $f \in C(\rz^2/(\ell_x\gz\oplus\ell_\xi\gz))$ be a Weyl
symbol with quantisation $\op_N(f)$. Then there exists a (generally
$N$-dependent) symbol $g$ that is a finite trigonometric polynomial,
such that
\begin{equation} 
  \op_N(f) = \op_N(g) \, .
\end{equation}
\end{lemma}
%
We remark that a similar statement can be found in
\cite[Theorem~1]{Lig16}.
\begin{proof}
We explicitly construct $g$. First observe the following periodicity
property of the phase space translations,
\begin{equation}
  T^{m+\mu N,n+\nu N} = (-1)^{m\nu+n\mu+\mu\nu N} \, T^{m,n} \, , 
\end{equation}
which follows from \eqref{eq:translatmatrix} by a direct computation.
Hence, 
\begin{equation}
\begin{split}
  \op_N(f) 
  &= \sum_{m,n\in\gz} f_{m,n} T^{m,n}
   = \sum_{m,n=0}^{N-1} \, \sum_{\mu,\nu\in\gz} 
     f_{m+\mu N,n+\nu N} T^{m+\mu N,n+\nu N} \\
  &= \sum_{m,n=0}^{N-1} \, \sum_{\mu,\nu\in\gz} 
     f_{m+\mu N,n+\nu N} \, (-1)^{m\nu+n\mu+\mu\nu N} \, T^{m,n} \, .
\end{split}
\end{equation}   
Defining 
\begin{equation}
  g_{m,n} 
  = \sum_{\mu,\nu\in\gz} f_{m+\mu N,n+\nu N} \, (-1)^{m\nu+n\mu+\mu\nu N} 
  \, , \quad 0 \leq m,n < N-1 \, , 
\end{equation}
concludes the proof.
\end{proof}
We now use this result to determine the matrix elements of
$\op_N(h)$, see~\eqref{q34}. To this end we first choose $f(x,\xi) =
a(\xi)/2$, with $a$ defined in Eq.~\eqref{eq:symbol_xi^2}, for the
symbol in Lemma~\ref{lemma:finite_symbols}. Then $f_{m,n}=a_m
\delta_{n0}$ and, likewise, $g_{m,n}=0$ for all $n\neq0$. For $m\neq0$
we have, cf.~\eqref{eq:Fourier_xi^2},
\begin{equation}
\begin{split}
\label{eq:gm0}
  g_{m,0} 
  &= \frac{\ell_\xi^2}{4\pi^2} \sum_{\mu\in\gz} 
     \frac{(-1)^{m+\mu N}}{(m+\mu N)^2}
   = \frac{\ell_\xi^2 (-1)^m}{4\pi^2N^2} \sum_{\mu\in\gz} 
     \frac{(-1)^{\mu N}}{(\mu + \frac{m}{N})^2}\\
  &= \frac{\ell_\xi^2 (-1)^m}{4N^2\sin^2\!\left(\pi\frac{m}{N}\right)}  
     \left\{ \begin{matrix} 
       1 & \text{if $N$ even} \\[.5ex] 
       \cos\left(\pi\frac{m}{N}\right) & \text{if $N$ odd} \end{matrix}
     \right. \, ,
\end{split}
\end{equation}
where we have used \cite[Eq.~5.15.6]{DLMF} to evaluate the
sums. The remaining coefficient is
\begin{equation}
\begin{split}
\label{eq:g00}
  g_{0,0} 
  = \frac{\ell_\xi^2}{24} + \frac{\ell_\xi^2}{2\pi^2} 
    \sum_{\mu=1}^\infty \frac{(-1)^{\mu N}}{\mu^2 N^2}
  &= \frac{\ell_\xi^2}{24N^2} 
     \left\{ \begin{matrix} 
       (N^2+2) & \text{if $N$ even} \\[1ex] 
       (N^2-1) & \text{if $N$ odd} \end{matrix}
     \right. \ .
\end{split}
\end{equation}

The second term of the symbol~\eqref{eq:hamiltonper}, $h-a/2=b/2$,
cf.~\eqref{eq:symbol_x^2}, is a function of $x$ only, and thus its
quantisation is diagonal in the chosen representation,
\begin{equation}
  \big( \op_N(b/2)\psi\big)_l 
  = -\frac{1}{2} \left(\frac{l}{N}\ell_x\right)^2 \psi_l \, .
\end{equation} 

Finally, the matrix elements of $\op_N(h)$ read
\begin{equation}
\label{eq:matrix_elements_xi^2-x^2}
  \op_N(h)_{k,l} 
  = \left( g_{0,0}-\frac{1}{2}\left(\frac{k}{N}\ell_x\right)^{\!2}\,\right) 
    \delta_{k,l} 
    +\sum_{m=1}^{N-1} g_{m,0} \, \delta_{k+m,l} \, , 
\end{equation}
with the coefficients $g_{m,0}$ given in Eqs.~\eqref{eq:gm0} and
\eqref{eq:g00}.

{\small \bibliographystyle{amsalpha} \bibliography{literature}}
\end{document}